\documentclass[a4paper,twoside,12pt,reqno]{article}
\usepackage[utf8]{inputenc}
\usepackage[english]{babel}
\usepackage{graphicx}
\usepackage{multirow}
\usepackage{epstopdf}
\usepackage[T1]{fontenc}
\usepackage{amssymb,amsthm,dsfont,mathrsfs}
\usepackage{amsfonts,euscript}
\usepackage{color}
\usepackage{authblk}
\usepackage[fleqn]{amsmath}
\usepackage[dvipsnames]{xcolor}
\usepackage{mathtools}
\usepackage{fullpage}
\usepackage{setspace}
\usepackage{hyperref}
\newtheorem{defi}{Definition}[section]
\newtheorem{teo}{Theorem}[section]
\newtheorem{pro}[teo]{Proposition}

\newtheorem{ex}{Example}[section]
\begin{document}

\renewcommand{\baselinestretch}{1.5}	
	
	\title{Multivariate Medial Correlation with applications}

		\author{Helena Ferreira}
		\affil{Universidade da Beira Interior, Centro de Matem\'{a}tica e Aplica\c{c}\~oes (CMA-UBI), Avenida Marqu\^es d'Avila e Bolama, 6200-001 Covilh\~a, Portugal, \texttt{helenaf@ubi.pt}}

		\author{Marta Ferreira}
		\affil{Center of Mathematics of Minho University, Center for Computational and Stochastic Mathematics of University of Lisbon,
			Center of Statistics and Applications of University of Lisbon, Portugal, \texttt{msferreira@math.uminho.pt} }

\date{}
	
	\maketitle

\abstract{We define a multivariate medial correlation coefficient that extends the probabilistic interpretation and properties of Blomqvist's $\beta$ coefficient, incorporates multivariate marginal dependencies and it preserves a stronger multivariate concordance relation. We determine the maximum and minimum values attainable and illustrate the results in some models. We end with an application on real datasets.}

\bigskip

\noindent\textbf{keywords:} {Blomqvist $\beta$, multivariate medial correlation, multivariate concordance measure}\\

\noindent\textbf{MSC2020 Subject Classification}: 62H20\\

%-------------------------------------------INTRODUÇÃO--------------
\section{Introduction}\label{sintro}

Let us consider that ${\bf{X}}=(X_1,X_2)$ is a real random vector, over the probability space $(\Omega, {\cal{A}}, P)$, with continuous marginal distribution functions $F_{X_i}$, $i=1,2$, and let $(U_1,U_2)$ represent the corresponding uniformized vector, that is, $U_i= F_{X_i}(X_i)$, $i=1,2$.

The medial correlation coefficient of $(X_1,X_2)$, which we will represent by $\beta(X_1,X_2)$ or $\beta({\bf{X}})$, is defined by
\begin{eqnarray}\label{betabiv}
\beta(X_1,X_2)=P\left(\left(U_1-\frac{1}{2}\right)\left(U_1-\frac{1}{2}\right)>0\right)-P\left(\left(U_1-\frac{1}{2}\right)\left(U_1-\frac{1}{2}\right)<0\right).
\end{eqnarray}

The $\beta$ coefficient introduced by Blomqvist (\cite{Blom}), has its value in $ [- 1,1] $ and compares the propensity for the margins of $(X_1,X_2)$ to take both values above or both values below their respective medians, with the propensity for the occurrence of the contrary event.

Since  
\begin{eqnarray}\label{betabivcopulas2}
\beta(X_1,X_2) =2\left(P\left(U_1>\frac{1}{2},U_2>\frac{1}{2}\right)+P\left(U_1<\frac{1}{2},U_2<\frac{1}{2}\right)\right)-1,
\end{eqnarray}  
and
\begin{eqnarray}\label{betabivcopulas3}
\beta(X_1,X_2) =4P\left(U_1<\frac{1}{2},U_2<\frac{1}{2}\right)-1,
\end{eqnarray}  
if $C_{\bf{X}}(u_1,u_2)$ and ${\hat{C}}_{\bf{X}}(u_1,u_2)$, $(u_1,u_2)\in [0,1]^2$, represent the copula and the survival copula of ${\bf{X}}$ ( Nelsen \cite{Nelsen}), respectively, we can say that 
%from (\ref{betabivcopulas2}) and (\ref{betabivcopulas3})
\begin{eqnarray}\label{betabivcopulas4}
\beta(X_1,X_2) =2\left(C_{\bf{X}}\left(\frac{1}{2},\frac{1}{2}\right)+{\hat{C}}_{\bf{X}}\left(\frac{1}{2},\frac{1}{2}\right)\right)-1,
\end{eqnarray} 
and
\begin{eqnarray}\label{betabivcopulas5}
\beta(X_1,X_2) =4C_{\bf{X}}\left(\frac{1}{2},\frac{1}{2}\right)-1.
\end{eqnarray} 
The bivariate medial correlation coefficient $\beta(X_1,X_2)$ enables to compare  $C_{\bf{X}}(u_1,u_2)$ on 
 $Q_L\cup Q_U=\left[0,\frac{1}{2}\right]^2\,\, \cup\,\, \left]\frac{1}{2},1\right]^2$ with $C_{\bf{X}}(u_1,u_2)$ on
$[0,1]^2\setminus \left(Q_L\cup Q_U\right)$ or to compare  $C_{\bf{X}}(u_1,u_2)$ on $Q_L=\left[0,\frac{1}{2}\right]^2$ with $C_{\bf{X}}(u_1,u_2)$ on $[0,1]^2\setminus Q_L$. 

The medial correlation coefficient can be related to other summary measures of dependence in $(X_1,X_2)$, or in $C_{\bf{X}}$, such as  Spearman’s $\rho$ or Kendall’s $\tau$ ( Nelsen \cite{Nelsen}, Joe \cite{Joe1}, Lebedev \cite{Leb} and references therein).

Two bivariate vectors ${\bf{X}}$ and ${\bf{Y}}$, or their copulas, can be partially ordered by punctually comparing their copulas. We say that ${\bf{X}}$ is less concordant than  ${\bf{Y}}$, and we write for that  ${\bf{X}}{\prec}_{c} {\bf{Y}}$, if $C_{\bf{X}}(u_1,u_2)\leq C_{\bf{Y}}(u_1,u_2)$,  $(u_1,u_2)\in [0,1]^2$, or equivalent, if  ${\hat{C}}_{\bf{X}}(u_1,u_2)\leq {\hat{C}}_{\bf{Y}}(u_1,u_2)$,  $(u_1,u_2)\in [0,1]^2$ (Nelsen  \cite{Nelsen}).

Thus, from the representations (\ref{betabivcopulas4}) or (\ref{betabivcopulas5}), we verify that 
\begin{eqnarray}\label{prec}
{\textrm{if}}\,\,\,  {\bf{X}}{\prec}_{c} {\bf{Y}}\,\,\,  {\textrm{then}}\,\,\,  \beta({\bf{X}})\leq \beta({\bf{Y}}).
\end{eqnarray}

 In addition to the increasing with concordance ordering, the bivariate medial correlation coefficient $\beta$ satisfies other properties that shape  the definition of measure of concordance according to Scarsini (\cite{Scarsini}).

Considering the product and minimum copulas, respectively, $C_{\Pi}(u_1,u_2)=u_1u_2$ and 
$C_{M}(u_1,u_2)=u_1\wedge u_2$, $(u_1,u_2)\in [0,1]^2$, we have $C_{\Pi} \prec_{c} C_{\bf{X}} \prec_{c} C_{M}$, $\beta(C_{\Pi})=0$, $\beta(C_{M})=1$ and we can also represent $\beta(X_1,X_2)$ by
\begin{eqnarray}\label{betabivcopulas-dif}
\beta(X_1,X_2) =2\left(C_{\bf{X}}\left(\frac{1}{2},\frac{1}{2}\right)-C_{\Pi}\left(\frac{1}{2},\frac{1}{2}\right)+{\hat{C}}_{\bf{X}}\left(\frac{1}{2},\frac{1}{2}\right)-{\hat{C}}_{\Pi}\left(\frac{1}{2},\frac{1}{2}\right)\right).
\end{eqnarray} 

For a random vector ${\bf{X}}=(X_1,...,X_d)$ with dimension $d>2$, if we think about generalizing   (\ref{betabiv}) to  $P\left(\displaystyle\prod_{i=1}^{d}\left(U_i-\frac{1}{2}\right)>0\right)-P\left(\displaystyle\prod_{i=1}^{d}\left(U_i-\frac{1}{2}\right)<0\right)$ we definitely loose:\\
(i) interpretation as a measure of propensity for all margins to exceed their respective medians or all margins to be below their medians, and\\
(ii) information about the behaviour of $C_{\bf{X}}$ on  $Q_{k}=\displaystyle\prod_{j=1}^{d}I_j$, $k=1,...,d-1$,  where $I_j=\left[0,\frac{1}{2}\right]$ for $k$ or $d-k$ values of $j$ and $I_j=\left]\frac{1}{2},1\right]$ for the others.\\
On the other hand, any generalization of $\beta$ in the multivariate context must preserve at least the property (i) and also verify\\ 
(iii) $\beta(C_{\Pi})=0$ and $\beta(C_{M})=1$.

The proposals of Nelsen (\cite{Nelsen2}), Úbeda-Flores (\cite{U-F}) and  Schmid and Schmidt (\cite{S+S}) manage to keep (i) and (iii) above.

Starting from the multivariate version of (\ref{betabivcopulas5}), $4C_{\bf{X}}(\frac{1}{2},...,\frac{1}{2})-1$,  rescaled by considering the quotient between its distance to the corresponding value for  $C_{\Pi}$ and the maximum value of that distance,
\begin{eqnarray}\label{Nelsenbetamulticopulas}
\begin{array}{ll}
\vspace{0,5cm}
\displaystyle\beta'(X_1,...,X_d) &=\displaystyle\frac{4C_{\bf{X}}\left(\frac{1}{2},...,\frac{1}{2}\right)-1- \left(4\left(\frac{1}{2}\right)^{d}-1\right)}{4C_{M}\left(\frac{1}{2},...,\frac{1}{2}\right)-1- \left(4\left(\frac{1}{2}\right)^{d}-1\right)}\\
&=\displaystyle\frac{2^{d}C_{\bf{X}}\left(\frac{1}{2},...,\frac{1}{2}\right)-1}{2^{d-1}-1},
\end{array}
\end{eqnarray} 
we find Nelsen's generalization  (\cite{Nelsen2}).

Úbeda-Flores (\cite{U-F}) proposes the extension of (\ref{betabivcopulas4}) in
\begin{eqnarray}\label{partida}
\begin{array}{ll}
2\left(C_{\bf{X}}\left(\frac{1}{2},...,\frac{1}{2}\right)+{\hat{C}}_{\bf{X}}\left(\frac{1}{2},...,\frac{1}{2}\right)\right)-1,
\end{array}
\end{eqnarray} 
also rescaled by considering the quotient between its distance to the corresponding value for  $C_{\Pi}$ and the maximum value of that distance. In this way, we obtain the following  generalization of $\beta$, which we will denote by $\beta^*$ and  where $\frac{\bf{1}}{\bf{2}}$ represents the vector of suitable size and coordinates all equal to $\frac{1}{2}$:

\begin{eqnarray}\label{H-Fbetamulticopulas}
\begin{array}{ll}
\vspace{0,5cm}
\displaystyle\beta^*(X_1,...,X_d) &=\displaystyle\frac{2\left(C_{\bf{X}}\left(\frac{1}{2},...,\frac{1}{2}\right)+{\hat{C}}_{\bf{X}}\left(\frac{1}{2},...,\frac{1}{2}\right)\right)-1- \left(\frac{1}{2^{d-2}}-1\right)}{2\left(C_{M}\left(\frac{1}{2},...,\frac{1}{2}\right)+{\hat{C}}_{M}\left(\frac{1}{2},...,\frac{1}{2}\right)\right)-1- \left(\frac{1}{2^{d-2}}-1\right)}\\
&=\displaystyle\frac{2^{d-1}\left(C_{\bf{X}}\left(\frac{\bf{1}}{\bf{2}}\right)+{\hat{C}}_{\bf{X}}\left(\frac{\bf{1}}{\bf{2}}\right)\right)-1}{2^{d-1}-1},
\end{array}
\end{eqnarray} 
which coincides with (\ref{Nelsenbetamulticopulas}) when $C={\hat{C}}$.

Reasoning in an equivalent way about (\ref{betabivcopulas-dif}), Schmid and Schmidt (\cite{S+S}) propose 
\begin{eqnarray}\nonumber%\label{S+Sbetamulticopulas}
\begin{array}{ll}
\displaystyle\frac{2\left(C_{\bf{X}}\left(\frac{\bf{1}}{\bf{2}}\right)-C_{\Pi}\left(\frac{\bf{1}}{\bf{2}}\right)+{\hat{C}}_{\bf{X}}\left(\frac{\bf{1}}{\bf{2}}\right)-{\hat{C}}_{\Pi}\left(\frac{\bf{1}}{\bf{2}}\right)\right)}{2\left(C_{M}\left(\frac{\bf{1}}{\bf{2}}\right)-C_{\Pi}\left(\frac{\bf{1}}{\bf{2}}\right)+{\hat{C}}_{M}\left(\frac{\bf{1}}{\bf{2}}\right)-{\hat{C}}_{\Pi}\left(\frac{\bf{1}}{\bf{2}}\right)\right)}&=
\displaystyle\frac{2^{d-1}\left(C_{\bf{X}}\left(\frac{\bf{1}}{\bf{2}}\right)+{\hat{C}}_{\bf{X}}\left(\frac{\bf{1}}{\bf{2}}\right)\right)-1}{2^{d-1}-1},
\end{array}
\end{eqnarray}
finding again the expression of Úbeda-Flores (\cite{U-F}). In addition to this extension, Schmid and Schmidt (\cite{S+S}) make a detailed study of a function resulting from a rescaling of  $C_{\bf{X}}({\bf{u}})+{\hat{C}}_{\bf{X}}({\bf{v}})$, ${\bf{u}}, {\bf{v}} \in [0,1]^d$, putting emphasis on the tail regions of the copula which determine the degree of large co-movements between the marginal random variables.

In order to keep (i), (ii) and (iii), we have Joe's sophisticated proposal (\cite{Joe2}) with an axiomatic on linear combinations of $\displaystyle C_{\sigma_{i_1}\sigma_{i_2}...\sigma_{i_k}\bf{X}}\left(\frac{\bf{1}}{\bf{2}}\right)$ and $\displaystyle \hat{C}_{\sigma_{i_1}\sigma_{i_2}...\sigma_{i_k}\bf{X}}\left(\frac{\bf{1}}{\bf{2}}\right)$, $1\leq i_1<...<i_K\leq d$, $k=[\frac{d+1}{2}],...,d$, where $\sigma_{j}\bf{X}$ denotes the j-th reflection of $\bf{X}$, that is, the vector $(X_1,...X_{j-1},-X_j,X_{j+1},...,X_d)$. Joe's axiomatic definition allows for various extensions of $\beta$, including those mentioned above and the arithmetic mean of
 $\beta(X_i,X_j)$, $1\leq i<j\leq d$. 

The extensions referred for $\beta$ increase with the  multivariate concordance (Joe \cite {Joe3}). We say that ${\bf{X}}=(X_1,...,X_d)$ is less concordant than ${\bf{Y}}=(Y_1,...,Y_d)$, or $C_{\bf{X}}$ is less concordant than  $C_{\bf{Y}}$, and in this case we write  ${\bf{X}}\prec_{c} {\bf{Y}}$, when we have
\begin{eqnarray}\label{concordancia}
\begin{array}{ll}
C_{\bf{X}}({\bf{u}})\leq C_{\bf{Y}}({\bf{u}})\,\,\,
\textrm{and}\,\,\,
\hat{C}_{\bf{X}}({\bf{u}})\leq \hat{C}_{\bf{Y}}(\bf{u}),
\end{array}
\end{eqnarray} 
for $\bf{u} \in [0,1]^d$. In the case of $d=2$ the two conditions are equivalent, as we have already mentioned.\\

The above proposed generalizations start from extensions of the representations of bivariate $\beta$ in terms of copulas, considering the corresponding multivariate copulas.\\
The proposal that we will make, in the next section, for a multivariate correlation coefficient $\beta(\bf{X})$ starts from  a generalization of the probabilistic interpretation of the definition (\ref{betabiv}) and satisfies almost all the desirable properties for a multivariate concordance measure (Taylor \cite{Taylor1},\cite{Taylor2}). It preserves a stronger multivariate concordance relation that we introduce in section \ref{propriedcoeficiente}. We present several representations for $\beta(\bf{X})$, we demonstrate the main properties, relate it to the previously mentioned coefficients and illustrate with examples and applications.

\section{Motivation for the multivariate medial correlation coefficient}\label{coeficientesmulti}

For $d\geq 2$, $D=\{1,...,d\}$, $I\subset D$, ${\bf{X}}=(X_1,...,X_d)$ with continuous marginal distributions and ${\bf{U}}=(U_1,...,X_d)=\left(F_{X_1}(X_1),...,F_{X_d}(X_d)\right)$, we define
\begin{eqnarray}\label{M-W}
\begin{array}{ll}
M(I)=\displaystyle\bigvee_{i\in I} U_i\,\,\,
\textrm{and}\,\,\,
W(I)=\displaystyle\bigwedge_{i\in I} U_i,
\end{array}
\end{eqnarray}
where  $\vee $ and $\wedge$ are the notations for the maximum and minimum operators, respectively.

When further clarification is needed, we write $M_{\bf{X}}(I)$ and $W_{\bf{X}}(I)$. 
Inequalities between vectors are understood by corresponding inequalities between homologous coordinates. By ${\bf{X}}_I$ we understand the subvector of ${\bf{X}}$ with margins in $I$ and ${\cal{P}}(D)$ represents the family of subsets of $D$.

Let's fix disjoint $I$ and $J$ in  ${\cal{P}}(D)$. The propensity for margins of ${\bf{X}}_I$ and margins of ${\bf{X}}_J$ simultaneously taking values below the respective medians or simultaneously values above the respective medians is evaluated by $C_{{\bf{X}}_{I\cup J}}(\frac{\bf{1}}{\bf{2}})+\hat{C}_{{\bf{X}}_{I\cup J}}(\frac{\bf{1}}{\bf{2}})$, that is, the probability of ${\bf{U}}_{I\cup J}$ taking values in $\left[0,\frac{1}{2}\right]^{|I\cup J|}\,\cup \,\left]\frac{1}{2},1\right]^{|I\cup J|}$. If we want to compare this probability with the probability of $\displaystyle{\bf{U}}_{I\cup J}$ taking values in $[0,1]^{|I\cup J|}\setminus \left(\left[0,\frac{1}{2}\right]^{|I\cup J|}\,\cup \, \left]\frac{1}{2},1\right]^{|I\cup J|}\right)$, we can do it briefly by calculating the coefficients
\begin{eqnarray}\label{beta de maximos}
\begin{array}{ll}
\vspace{0,5cm}
&\beta\left(M(I),M(J)\right):=\\
:=&P\left(\left(M(I)-\frac{1}{2}\right)\left(M(J)-\frac{1}{2}\right)>0\right)-P\left(\left(M(I)-\frac{1}{2}\right)\left(M(J)-\frac{1}{2}\right)<0\right)\\
=&2\left(P\left(M(I)>\frac{1}{2},M(J)>\frac{1}{2}\right)+P\left(M(I)< \frac{1}{2},M(J)< \frac{1}{2}\right)\right)-1
\end{array}
\end{eqnarray}
and 
\begin{eqnarray}\label{beta de minimos}
\begin{array}{ll}
\vspace{0,5cm}
&\beta(W(I),W(J)):=\\
:=&P\left(\left(W(I)-\frac{1}{2}\right)\left(W(J)-\frac{1}{2}\right)>0\right)-P\left(\left(W(I)-\frac{1}{2}\right)\left(W(J)-\frac{1}{2}\right)<0\right)\\
=&2\left(P\left(W(I)>\frac{1}{2},W(J)>\frac{1}{2}\right)+P\left(W(I)< \frac{1}{2},W(J)< \frac{1}{2}\right)\right)-1.
\end{array}
\end{eqnarray}

\vspace{0,5cm}
Let us make some comments about
\begin{eqnarray}\label{mediadebetas}
\begin{array}{ll}
\displaystyle\beta_{I,J}({\bf{X}}):=\displaystyle\frac{\beta(M(I),M(J))+\beta(W(I),W(J))}{2}.
\end{array}
\end{eqnarray}

(i) The expressions (\ref{beta de maximos}), (\ref{beta de minimos}) and (\ref{mediadebetas}) have $\beta(X_i,X_j)$ as a particular case,  if we take $I=\{i\}$ and $J=\{j\}$.\\
If $I=D$, $J=\emptyset$ and we consider that $M(\emptyset)=-\infty$ and $W(\emptyset)=+\infty$, then (\ref{mediadebetas}) is equal to $C_{\bf{X}}\left(\frac{\bf{1}}{\bf{2}}\right)+{\hat{C}}_{\bf{X}}\left(\frac{\bf{1}}{\bf{2}}\right)-1$, which can be rescaled in order to obtain  the proposal of Úbeda-Flores (\cite{U-F}) and Schmid and Schmidt (\cite{S+S}).\\

(ii) Since $\displaystyle\beta_{I,J}({\bf{X}})$ is defined as an average of bivariate coefficients, it can be estimated by the methods available for the bivariate context (Blomqvist \cite {Blom}, Schmid and Schmidt \cite{S+S} and references therein).

%(iii) From  (\ref{prec}), the value of (\ref{mediadebetas}) increases with the concordances of $(M(I),M(J))$ and $(W(I),W(J))$, that is, for two vectors of dimension $d$, ${\bf{X}}$ and ${\bf{Y}}$, if 
%\begin{eqnarray}\label{concordanciaM-W}
%\begin{array}{ll}
%P\left(M_{\bf{X}}(I)\leq u, M_{\bf{X}}(J)\leq v\right)&\leq P\left(M_{\bf{Y}}(I)\leq u, M_{\bf{Y}}(J)\leq v\right),\\
%P\left(W_{\bf{X}}(I)> u, W_{\bf{X}}(J)> v\right)&\leq P\left(W_{\bf{Y}}(I)>u, W_{\bf{Y}}(J)> v\right),
%\end{array} 
%\end{eqnarray}
%$(u,v)\in [0,1]^2$, then $\beta_{I,J}({\bf{X}})\leq \beta_{I,J}({\bf{Y}})$. 
%
%Note that the first inequality is equivalent to the similar inequality for the survival functions of
%$(M_{\bf{.}}(I), M_{\bf{.}}(J))$ and the second is equivalent to the similar inequality for the distribution functions of $(W_{\bf{.}}(I), W_{\bf{.}}(J))$, since we are in the context of bivariate vectors.
%
%(iv) Since (\ref{concordanciaM-W}) is implied by ${\bf{X}}\prec_{c} {\bf{Y}}$, we can say that $\beta_{I,J}({\bf{X}})$  increases with the concordance of ${\bf{X}}$.

(iii) If $C_{{\bf{X}}}=C_{M}$ we have $\beta_{I,J}({\bf{X}})=1$ and if $C_{{\bf{X}}}=C_{\Pi}$ then $\beta_{I,J}({\bf{X}})=2^{2-|I|-|J|}-2^{1-|I|}-2^{1-|J|}+1=(2^{1-|I|}-1)(2^{1-|J|}-1)$, where $|A|$ denotes the cardinality of $A$. This value becomes null if and only if $|I|=1$ or $|J|=1$.

(iv) A linear combination of $\beta_{\{i\},\{j\}}({\bf{X}})$, $1\leq i<j\leq d$, takes into account the bivariate dependencies in ${\bf{X}}$, but if we consider some function of the coefficients $\beta_{I,J}({\bf{X}})$, with $I,J \in   {\cal{F}}$, for some family ${\cal{F}}\subset {\cal{P}}(D)$ containing sets with more than one element, then we will be incorporating multivariate marginal dependencies.\\

The definition we propose, in the next section, for a multivariate medial correlation coefficient, will be based on the bivariate coefficients $\beta_{\{i\},D\setminus \{i\}}({\bf{X}})$, $1\leq i\leq d$, incorporating the dependency between each margin $X_i$ and ${\bf{X}}_{D\setminus \{i\}}$, $1\leq i\leq d$.\\
Our proposal contains, as a particular case, the Blomqvist bivariate coefficient, extends the probabilistic interpretation (\ref{betabiv}), takes values in $[-1,1]$, becoming null naturally when $C_{\bf{X}}=C_{\Pi}$ and taking the maximum value when $C_{\bf{X}}=C_{M}$. The rest of the properties we proved allow us to consider it a measure for a  multivariate concordance relation stronger than concordance order.

\section{A multivariate medial correlation coefficient}\label{defcoeficiente}

\begin{defi}\label{defcoefmulti} The multivariate medial correlation coefficient of the vector ${\bf{X}}$ with dimension $d$, or of its copula $C_{\bf{X}}$, is defined as
\begin{eqnarray}\label{defbetamulti}
\beta({\bf{X}})=\displaystyle\frac{1}{d}\sum_{i=1}^d\beta_{\{i\},D\setminus \{i\}}({\bf{X}}),
\end{eqnarray}
where
\begin{eqnarray}\label{defbeta,i,D-i}
\displaystyle\beta_{\{i\},D\setminus \{i\}}({\bf{X}})=\displaystyle\frac{\beta\left(U_i,M(D\setminus \{i\})\right)+\beta\left(U_i,W(D\setminus \{i\})\right)}{2},\,\,\,i=1,...,d.
\end{eqnarray}
\end{defi}

%In what follows, for the sake of simplicity of writing, we sometimes write $D\setminus \{i\}$ rather than $D\setminus \{i\}$, $D_{\{i,j\}}$ for $D\setminus \{i,j\}$ and $\beta_{i,D\setminus \{i\}}({\bf{X}})$ for $\beta_{\{i\},D\setminus \{i\}}({\bf{X}})$.\\
Below we present some representations of $\beta({\bf{X}})$ that will be useful to clarify their properties and interpretation. The following
\begin{eqnarray}\label{beta,i,D-i}
\begin{array}{ll}
\beta_{i,D\setminus \{i\}}({\bf{X}})&=2\left(P\left(U_i<\frac{1}{2},M\left(D\setminus \{i\}\right)<\frac{1}{2}\right)+P\left(U_i>\frac{1}{2},W\left(D\setminus \{i\}\right)>\frac{1}{2}\right)\right)\\
&-P\left(M\left(D\setminus \{i\}\right)<\frac{1}{2})-P(W\left(D\setminus \{i\}\right)>\frac{1}{2}\right),
\end{array} 
\end{eqnarray}
holds, generalizing (\ref{betabivcopulas2}).
We also have
\begin{eqnarray}\label{beta,i,D-i-cópula}
\begin{array}{ll}
\beta_{i,D\setminus \{i\}}({\bf{X}})=2\left(C_{\bf{X}}\left(\frac{\bf{1}}{\bf{2}}\right)+{\hat{C}}_{\bf{X}}\left(\frac{\bf{1}}{\bf{2}}\right)\right)-C_{\bf{X}_{D\setminus \{i\}}}\left(\frac{\bf{1}}{\bf{2}}\right)-{\hat{C}}_{\bf{X}_{D\setminus \{i\}}}\left(\frac{\bf{1}}{\bf{2}}\right),
\end{array} 
\end{eqnarray}
generalizing  (\ref{betabivcopulas4}).
From the previous relation, it follows that
\begin{eqnarray}\label{beta,i,D-i-cópula2}
\begin{array}{ll}
\beta_{i,D\setminus \{i\}}({\bf{X}})=C_{\bf{X}}\left(\frac{\bf{1}}{\bf{2}}\right)+{\hat{C}}_{\bf{X}}\left(\frac{\bf{1}}{\bf{2}}\right)-C_{\sigma_i\bf{X}}\left(\frac{\bf{1}}{\bf{2}}\right)-{\hat{C}}_{\sigma_i\bf{X}}\left(\frac{\bf{1}}{\bf{2}}\right),
\end{array} 
\end{eqnarray}
where $\sigma_i\bf{X}$ is the i-th reflection of $\bf{X}$, that is, $\sigma_i{\bf{X}}=(X_1,....,X_{i-1},-X_i,X_{i+1},...,X_d)$ and therefore $C_{\sigma_i{\bf{X}}}(\frac{\bf{1}}{\bf{2}})=C_{(U_1,...,U_{i-1},1-U_i,U_{i+1},....,U_d)}(\frac{\bf{1}}{\bf{2}})$.
We then obtain the following ways of representing the coefficient $\beta$.

\begin{pro}\label{representaçoesBeta}
	The multivariate medial correlation coefficient of the vector ${\bf{X}}$ with dimension $d$,  admits the following representations:
\begin{eqnarray}\label{R1}
\begin{array}{rl}
 \beta({\bf{X}})=&2\left(P\left({\bf{U}}\leq \frac{\bf{1}}{\bf{2}}\right)+P\left({\bf{U}}> \frac{\bf{1}}{\bf{2}}\right)\right)\\
 &-\displaystyle\frac{1}{d}\sum_{i=1}^d\left(P\left({\bf{U}}_{D\setminus \{i\}}\leq \frac{\bf{1}}{\bf{2}}\right)+P\left({\bf{U}}_{D\setminus \{i\}}> \frac{\bf{1}}{\bf{2}}\right)\right),
 \end{array}
\end{eqnarray} 
\begin{eqnarray}\label{R2}
 \beta({\bf{X}})=2\left(C_{\bf{X}}\left(\frac{\bf{1}}{\bf{2}}\right)+{\hat{C}}_{\bf{X}}\left(\frac{\bf{1}}{\bf{2}}\right)\right)-\displaystyle\frac{1}{d}\sum_{i=1}^d\left(C_{\bf{X}_{D\setminus \{i\}}}\left(\frac{\bf{1}}{\bf{2}}\right)+{\hat{C}}_{\bf{X}_{D\setminus \{i\}}}\left(\frac{\bf{1}}{\bf{2}}\right)\right),
\end{eqnarray}
\begin{eqnarray}\label{R3}
\beta({\bf{X}})=C_{\bf{X}}\left(\frac{\bf{1}}{\bf{2}}\right)+{\hat{C}}_{\bf{X}}\left(\frac{\bf{1}}{\bf{2}}\right)-\displaystyle\frac{1}{d}\sum_{i=1}^d\left(C_{\sigma_i\bf{X}}\left(\frac{\bf{1}}{\bf{2}}\right)+{\hat{C}}_{\sigma_i\bf{X}}\left(\frac{\bf{1}}{\bf{2}}\right)\right).
\end{eqnarray}
\end{pro}

The relation (\ref{R3}) rewritten in the form
\begin{eqnarray}\nonumber%\label{R3}
\beta({\bf{X}})=\displaystyle\frac{1}{d}\sum_{i=1}^d\left(C_{\bf{X}}\left(\frac{\bf{1}}{\bf{2}}\right)-C_{\sigma_i\bf{X}}\left(\frac{\bf{1}}{\bf{2}}\right)+{\hat{C}}_{\bf{X}}\left(\frac{\bf{1}}{\bf{2}}\right)-{\hat{C}}_{\sigma_i{\bf{X}}}\left(\frac{\bf{1}}{\bf{2}}\right)\right),
\end{eqnarray}
reinforces the idea that $\beta({\bf{X}})$ compares  the propensity of each margin $X_i$ to agree with the remaining margins together, ${\bf{X}}_{D\setminus \{i\}}$, and the propensity to disagree with them, when they are all above or all below their respective medians.

In the following, we establish relationships between $\beta({\bf{X}})$ and the generalizations referred to in the introduction.
By applying the definition (\ref{H-Fbetamulticopulas}) of $\beta^*$,  we conclude from the representation  ($\ref{R3})$ that
\begin{eqnarray}\nonumber
\begin{array}{ll}
\beta({\bf{X}})&=\displaystyle\frac{(2^{d-1}-1)\beta^*({\bf{X}})+1}{2^{d-1}}-\displaystyle\frac{1}{d}\sum_{i=1}^d\displaystyle\frac{(2^{d-1}-1)\beta^*(\sigma_i{\bf{X}})+1}{2^{d-1}}\\
&=\displaystyle\frac{(2^{d-1}-1)}{2^{d-1}}\left(\beta^*({\bf{X}})-\displaystyle\frac{1}{d}\sum_{i=1}^d\beta^*(\sigma_i{\bf{X}})\right).
\end{array}
\end{eqnarray}
By defining ${\bar{N}}=\displaystyle\sum_{i=1}^d\mathbf{1}_{\{U_i>\frac{1}{2}\}}$, the representation   (\ref{R3}) of $\beta$ leads to
\begin{eqnarray}\label{excedências}
\begin{array}{ll}
\beta({\bf{X}})=P({\bar{N}}=0)+P({\bar{N}}=d)-\displaystyle\frac{1}{d}\left(P({\bar{N}}=1)+P({\bar{N}}=d-1)\right).
\end{array}
\end{eqnarray}
That fits Joe's representation (3.1.1) (\cite{Joe2}) with $w_d=1$, $w_{d-1}=-\frac{1}{d}$ and the remaining weights  $w_i$ equal to zero.

Note that in the $3$-dimensional case, the multivariate medial correlation coefficient $\beta$ satisfies
\begin{eqnarray}\nonumber
\begin{array}{rl}
&\beta({\bf{X}})=\frac{4}{3}C_{\bf{X}}\left(\frac{\bf{1}}{\bf{2}}\right)+\frac{4}{3}\hat{C}_{\bf{X}}\left(\frac{\bf{1}}{\bf{2}}\right)-\frac{1}{3}=\displaystyle\beta^{*}({\bf{X}})=\frac{\beta(X_1,X_2)+\beta(X_1,X_3)+\beta(X_2,X_3)}{3}.
\end{array}
\end{eqnarray}
Thus, in the $3$-dimensional case $\beta$ equals $\beta^{*}$ and hence allows a different view on Blomqvist's $\beta$ discussed in Úbeda-Flores (\cite{U-F}).

We refer the properties of $\beta({\bf{X}})$ in the next section and end this one with three examples.

\begin{ex}\label{ex1}
	Consider $C_{\bf{X}}(u_1,...,u_4)=\left(u_1^{\delta}\wedge u_2\right)u_1^{1-\delta}\left(u_3^{\alpha}\wedge u_4\right)u_3^{1-\alpha}$, with $0\leq \delta, \alpha\leq 1$, that is, $C_{\bf{X}}$ is the product of two Marshall-Olkin survival copulas (\cite{Joe3}).
It holds that
\begin{eqnarray}\nonumber
\begin{array}{ll}
\vspace{0,5cm}
\displaystyle C_{\bf{X}}\left(\frac{\bf{1}}{\bf{2}}\right)={\hat{C}}_{\bf{X}}\left(\frac{\bf{1}}{\bf{2}}\right)=\displaystyle\left(\frac{1}{2}\right)^{4-\delta- \alpha},\\
\vspace{0,5cm}
\displaystyle C_{{\bf{X}}_{D\setminus{\{1\}}}}\left(\frac{\bf{1}}{\bf{2}}\right)={\hat{C}}_{{\bf{X}}_{D\setminus{\{1\}}}}\left(\frac{\bf{1}}{\bf{2}}\right)=C_{{\bf{X}}_{D\setminus{\{2\}}}}\left(\frac{\bf{1}}{\bf{2}}\right)={\hat{C}}_{{\bf{X}}_{D\setminus{\{2\}}}}\left(\frac{\bf{1}}{\bf{2}}\right)=\displaystyle\left(\frac{1}{2}\right)^{3- \alpha},\\

\displaystyle C_{{\bf{X}}_{D\setminus{\{3\}}}}\left(\frac{\bf{1}}{\bf{2}}\right)={\hat{C}}_{{\bf{X}}_{D\setminus{\{3\}}}}\left(\frac{\bf{1}}{\bf{2}}\right)=C_{{\bf{X}}_{D\setminus{\{4\}}}}\left(\frac{\bf{1}}{\bf{2}}\right)={\hat{C}}_{{\bf{X}}_{D\setminus{\{4\}}}}\left(\frac{\bf{1}}{\bf{2}}\right)=\displaystyle\left(\frac{1}{2}\right)^{3- \delta}.
\end{array}
\end{eqnarray}
Therefore,
\begin{eqnarray}\nonumber
\begin{array}{ll}
\beta({\bf{X}})=\displaystyle 2^{\delta+ \alpha-2}-2^{\alpha-3}-2^{\delta-3}.
\end{array}
\end{eqnarray}
In the case of $\delta=\alpha=0$ the result  agrees with what we expect, since in this case the margins of ${\bf{X}}$ are independent. The expression obtained can be related to $\beta(X_1,X_2)$ and $\beta(X_3,X_4)$
through
\begin{eqnarray}\nonumber
\begin{array}{ll}
\beta({\bf{X}})&=2 \times 2^{\delta+ \alpha-3}-2^{\alpha-3}-2^{\delta-3}=\left(2^{\delta+ \alpha-3}-2^{\alpha-3}\right)+\left(2^{\delta+ \alpha-3}-2^{\delta-3}\right)\\
&=2^{\alpha-3}\left(2^{\delta}-1\right)+2^{\delta-3}\left(2^{\alpha}-1\right)\\
&=2^{\alpha-3}\beta(X_1,X_2)+2^{\delta-3}\beta(X_3,X_4),
\end{array}
\end{eqnarray}
We verify that $\beta({\bf{X}})$ increases with $\delta$ and $\alpha$, generalizing what we already knew to
$\beta(X_1,X_2)$ and $\beta(X_3,X_4)$. Therefore $\beta({\bf{X}})$ increases with the concordance of  ${\bf{X}}$.\\
\end{ex}

\begin{ex}\label{ex2}
 Let us consider that ${\bf{X}}$ has a trivariate  Gumbel copula $C_{\bf{X}}(u_1,u_2,u_3)=\displaystyle\exp\left\{-\left(\sum_{i=1}^3\left(-\ln u_i\right)^{1/\delta}\right)^{\delta}\right\}$, with $0<\delta \leq 1$.
It holds that
\begin{eqnarray}\nonumber
\begin{array}{ll}
C_{\bf{X}}\left(\frac{\bf{1}}{\bf{2}}\right)=2^{-3^{\delta}},\,\,\,{\hat{C}}_{\bf{X}}\left(\frac{\bf{1}}{\bf{2}}\right)=3\times 2^{-2^{\delta}}-2^{-3^{\delta}}-2^{-1}
\end{array}
\end{eqnarray}
and
\begin{eqnarray}\nonumber
\begin{array}{ll}
C_{{\bf{X}}_{D\setminus \{i\}}}(\frac{\bf{1}}{\bf{2}})={\hat{C}}_{{\bf{X}}_{D\setminus \{i\}}}(\frac{\bf{1}}{\bf{2}})=2^{-2^{\delta}},\,\,\, \textrm{for}\,\,\, i=1,2,3.
\end{array}
\end{eqnarray}
Therefore, we obtain $\beta({\bf{X}})=2^{2-2^{\delta}}-1$, coincident with $\beta(X_i,X_j)$, $1\leq i<j\leq 3$.

With simple calculations we can also conclude that
$$\beta(-X_1,X_2,X_3)=\frac{-2^{2-2^{\delta}}+1}{3}$$ and  that
$$\beta(X_1,X_2,X_3)+\beta(-X_1,X_2,X_3)=\frac{2}{2+1}\beta(X_2,X_3),$$
 which corresponds to the verification in this example of a transition property that we present in the next section.
Before we present the general expression of the multivariate correlation coefficient for a  Gumbel distribution of dimension $d\geq 1$, let’s also calculate it specifically for $d=4$.

We have 
\begin{eqnarray}\nonumber
\begin{array}{ll}
C_{\bf{X}}\left(\frac{\bf{1}}{\bf{2}}\right)=2^{-4^{\delta}},\,\,\,{\hat{C}}_{\bf{X}}\left(\frac{\bf{1}}{\bf{2}}\right)=-1+6\times 2^{-2^{\delta}}-4\times 2^{-3^{\delta}}+2^{-4^{\delta}},
\end{array}
\end{eqnarray}
and
\begin{eqnarray}\nonumber
\begin{array}{ll}
C_{{\bf{X}}_{D\setminus \{i\}}}\left(\frac{\bf{1}}{\bf{2}}\right)=2^{-3^{\delta}},\,\,
{\hat{C}}_{{\bf{X}}_{D\setminus \{i\}}}\left(\frac{\bf{1}}{\bf{2}}\right)=3\times 2^{-2^{\delta}}-2^{-3^{\delta}}-2^{-1},\,\,\textrm{for}\,\, i=1,2,3.
\end{array}
\end{eqnarray}

Then  
$$\beta(X_1,X_2,X_3,X_4)=4\times 2^{-4^{\delta}}-8\times 2^{-3^{\delta}}+9\times 2^{-2^{\delta}}-\frac{3}{2}.$$
These results for $d=2,3,4$, calculated directly, can also be obtained from the following general result.

If $d$ is even, we have
$$\beta({\bf{X}})=\displaystyle\frac{1-d}{2}+\sum_{k=1}^{d-2} \left(\left(^{d-1}_{\,\,\,k}\right)+\left(^{\,\,\,d}_{k+1}\right)\right)(-1)^{k+1}2^{-(k+1)^{\delta}} +4\times 2^{-d^{\delta}}+(-1)^{d-1}2^{-(d-1)^{\delta}},$$
(considering that a sum with the initial value of the counter greater than the final one is null) and if $d$ is odd, we have
$$\beta({\bf{X}})=\displaystyle\frac{1-d}{2}+\sum_{k=1}^{d-2} \left(\left(^{d-1}_{\,\,\,k}\right)+\left(^{\,\,\,d}_{k+1}\right)\right)(-1)^{k+1}2^{-(k+1)^{\delta}} -2^{-(d-1)^{\delta}}.$$
\end{ex}

The third example also serves as a motivation for one of the properties in the next section, on the best lower limit of $\beta({\bf{X}})$.

\begin{ex}\label{ex3}
Consider ${\bf{X}}$ of dimension $d$ such that ${\bf{U}}=(U,1-U,U_3,...,U_d)$. Then 
\begin{eqnarray}\nonumber
\begin{array}{rl}
\vspace{0,5cm}
\beta({\bf{X}})=&2\times (0+0)\\
&-\displaystyle\frac{1}{d}\left(C_{{\bf{X}}_{D\setminus {\{1\}}}}\left(\frac{\bf{1}}{\bf{2}}\right)+{\hat{C}}_{{\bf{X}}_{D\setminus {\{1\}}}}\left(\frac{\bf{1}}{\bf{2}}\right)+C_{{\bf{X}}_{D\setminus {\{2\}}}}\left(\frac{\bf{1}}{\bf{2}}\right)+{\hat{C}}_{{\bf{X}}_{D\setminus {\{2\}}}}\left(\frac{\bf{1}}{\bf{2}}\right)+0\right)\\\\
\vspace{0.5cm}
=&-\displaystyle\frac{1}{d}\left(C_{{\bf{X}}_{D\setminus {\{1\}}}}\left(\frac{\bf{1}}{\bf{2}}\right)+C_{{\bf{X}}_{D\setminus {\{2\}}}}\left(\frac{\bf{1}}{\bf{2}}\right)+{\hat{C}}_{{\bf{X}}_{D\setminus {\{1\}}}}\left(\frac{\bf{1}}{\bf{2}}\right)+{\hat{C}}_{{\bf{X}}_{D\setminus {\{2\}}}}\left(\frac{\bf{1}}{\bf{2}}\right)\right)\\
\vspace{0.5cm}
=&-\displaystyle\frac{1}{d}\left(C_{{\bf{X}}_{D\setminus {\{1,2\}}}}\left(\frac{\bf{1}}{\bf{2}}\right)+{\hat{C}}_{{\bf{X}}_{D\setminus {\{1,2\}}}}\left(\frac{\bf{1}}{\bf{2}}\right)\right).
\end{array}
\end{eqnarray}
It follows that  $\beta({\bf{X}})\geq -\frac{1}{d}$ and if, in particular $(U_3,...,U_d)=(V,...,V)$, then $\beta({\bf{X}})= -\frac{1}{d}$.
\end{ex}

\section{Properties of the multivariate medial correlation coefficient}\label{propriedcoeficiente}

Since the  coefficients $\beta_{\{i\},D\setminus \{i\}}({\bf{X}})$, $i=1,...,d$, take values in $[-1,1]$, the proposed coefficient takes values in the same range, being null for $C_{{\bf{X}}}=C_{\Pi}$. The maximum value is attainable when $C_{{\bf{X}}}=C_{M}=1$ and the minimum attainable value is equal to  $-\frac{1}{d}$. In fact, from the representation (\ref{R3}), we verify that  $\beta({\bf{X}})$ takes the minimum value when $C_{{\bf{X}}}\left(\frac{\bf{1}}{\bf{2}}\right)+{\hat{C}}_{{\bf{X}}}\left(\frac{\bf{1}}{\bf{2}}\right)=0$ and $\displaystyle\sum_{i=1}^d\left(C_{\sigma_i\bf{X}}\left(\frac{\bf{1}}{\bf{2}}\right)+{\hat{C}}_{\sigma_i\bf{X}}\left(\frac{\bf{1}}{\bf{2}}\right)\right)=1$, what happens when, for example, $U_j=1-U_i$ for some pair $1\leq i<j\leq d$ and $U_k=V$ for each $k\in D\setminus \{i,j\}$, analogously to what we saw in the example  \ref{ex3}. 

The value of $\beta({\bf{X}})$ may not increase with the concordance of ${\bf{X}}$. We can verify this with an example proposed by an anonymous referee. 

Consider ${\bf{X}}$ and ${\bf{Y}}$ $4$-dimensional vetors with copulas, respectively,
\begin{eqnarray}\nonumber
\begin{array}{ll}
C_{\bf{X}}\left(u_1,u_2,u_3,u_4\right)=C_W(u_1,u_2)C_{\Pi}(u_3,u_4)
\end{array}
\end{eqnarray}
and
\begin{eqnarray}\nonumber
\begin{array}{ll}
C_{\bf{Y}}\left(u_1,u_2,u_3,u_4\right)=C_W(u_1,u_2)C_{M}(u_3,u_4),
\end{array}
\end{eqnarray}
where $C_{W}$ denotes the countermonotonicity copula, $C_{W}(u_1,u_2)=(u_1+u_2-1)\vee 0$. We have ${\bf{X}}{\prec}_{c} {\bf{Y}}$ and however $\beta({\bf{X}})=-\frac{1}{8}>-\frac{1}{4} =\beta({\bf{Y}})$.\\

If ${\bf{X}}{\prec}_{c} {\bf{Y}}$ and, for each $i\in D$,
\begin{eqnarray}\label{strongconcordance}
\left\{\begin{array}{ll}
C_{\sigma_i\bf{Y}}\left(\frac{\bf{1}}{\bf{2}}\right)\leq C_{\sigma_i\bf{X}}\left(\frac{\bf{1}}{\bf{2}}\right)\\
{\hat{C}}_{\sigma_i\bf{Y}}\left(\frac{\bf{1}}{\bf{2}}\right)\leq {\hat{C}}_{\sigma_i\bf{X}}\left(\frac{\bf{1}}{\bf{2}}\right)\,,\,\,i\in D
\end{array}\right.
\end{eqnarray}
then, from proposition \ref{representaçoesBeta}, (\ref{R3}), we can conclude that $\beta({\bf{X}})\leq \beta({\bf{Y}})$.\\

The verification of condition (\ref{strongconcordance}) together with ${\bf{X}}{\prec}_{c} {\bf{Y}}$, which can be illustrated with example  \ref{ex2}, tells us that, in addition to the propensity for all margins to exceed their respective medians or all margins to be below their medians to be higher in ${\bf{Y}}$, also the propensity for each margin to disagree with the remaining, in this sense, is lower in ${\bf{Y}}$, reinforcing the relation ${\bf{X}}{\prec}_{c} {\bf{Y}}$. 

When we have ${\bf{X}}{\prec}_{c} {\bf{Y}}$ and (\ref{strongconcordance}) we denote this type of relation by ${\bf{X}}{\prec}{\prec}_{c} {\bf{Y}}$.\\

The above properties on the values of the multivariate medial correlation coefficient are arranged in the following proposition.

\begin{pro}\label{valoresBeta}
	The values of the multivariate medial correlation coefficient for vectors of dimension $d$ satisfy the following properties:\\
(i) If ${\bf{X}}{\prec}{\prec}_{c} {\bf{Y}}$ then $\beta({\bf{X}})\leq \beta({\bf{Y}})$.\\
(ii) If $C_{\bf{X}}=C_{\Pi}$ then $\beta({\bf{X}})=0$.\\
(iii) If $C_{\bf{X}}=C_{W}$ then $\beta({\bf{X}})=1$.\\
(iv) The minimum attainable value for $\beta({\bf{X}})$ is $-\frac{1}{d}$.
\end{pro}

In the proposition below we present the properties of continuity, permutation invariance, duality,
reflection symmetry and transition, which together with (i)-(iii) of the previous proposition and following Taylor \cite{Taylor1}, \cite{Taylor2}, justifies  calling the proposed coefficient a measure for the concordance relation $\prec\prec_C$.

\begin{pro}\label{propriedadesBeta}
	The values of the multivariate medial correlation coefficient for vectors of dimension $d$ satisfy the following properties:\\
(i) If $\{C_{{\bf{X}}_n}\}_{n\geq 1}$ converges uniformly to $C_{\bf{X}}$, $n \to +\infty$, then $\displaystyle\lim_{n \to +\infty}\beta({\bf{X}}_n)=\beta({\bf{X}})$.\\
(ii) The value of  $\beta({\bf{X}})$ is invariant for permutations of the margins of ${\bf{X}}$.\\
(iii) $\beta({\bf{X}})=\beta(-{\bf{X}})$.\\
(iv) $\displaystyle\sum_{(\epsilon_1,...,\epsilon_d)\in \{-1,1\}^{d}}\beta(\epsilon_1X_1,...,\epsilon_dX_d)=0$.\\
(v) If ${\bf{Y}}$ is a $(d+1)$-dimensional random vector such that $C_{\bf{Y}}(u_1,...,u_{i-1},1,u_{i+1},...,u_d)=C_{\bf{X}}(u_1,...,u_{i-1},u_{i+1},...,u_d)$ then $\displaystyle\frac{d}{d+1}\beta({\bf{X}})=\beta({\bf{Y}})+\beta(\sigma_i{\bf{Y}}).$
\end{pro}
\begin{proof}
The statement of (i) can be obtained, for example, from (\ref{R2}).
From the representation  (\ref{excedências})  we can conclude (ii).
The representation (\ref{R3})  leads to (iii) and (iv).
Finally to obtain (v), let us note that, by (\ref{R3}), we have
\begin{eqnarray}\nonumber%\label{transição}
\begin{array}{rl}
\vspace{0,5cm}
&\beta({\bf{Y}})+\beta(\sigma_i{\bf{Y}})\\
=&C_{\bf{Y}}\left(\frac{\bf{1}}{\bf{2}}\right)+C_{\sigma_i\bf{Y}}\left(\frac{\bf{1}}{\bf{2}}\right)+{\hat{C}}_{\bf{Y}}\left(\frac{\bf{1}}{\bf{2}}\right)+{\hat{C}}_{\sigma_i\bf{Y}}\left(\frac{\bf{1}}{\bf{2}}\right)\\
\vspace{0,5cm}
&-\displaystyle\frac{1}{d+1}\left(C_{\sigma_i\bf{Y}}\left(\frac{\bf{1}}{\bf{2}}\right)+{\hat{C}}_{\sigma_i\bf{Y}}\left(\frac{\bf{1}}{\bf{2}}\right)+C_{\sigma_i\sigma_i\bf{Y}}\left(\frac{\bf{1}}{\bf{2}}\right)+{\hat{C}}_{\sigma_i\sigma_i\bf{Y}}\left(\frac{\bf{1}}{\bf{2}}\right)\right)\\
\vspace{0,5cm}
&-\displaystyle\frac{1}{d+1}\sum_{j=1,j\neq i}^{d+1}\left(C_{\sigma_j\bf{Y}}\left(\frac{\bf{1}}{\bf{2}}\right)+C_{\sigma_j\sigma_i\bf{Y}}\left(\frac{\bf{1}}{\bf{2}}\right)+{\hat{C}}_{\sigma_j\bf{Y}}\left(\frac{\bf{1}}{\bf{2}}\right)+{\hat{C}}_{\sigma_j\sigma_i\bf{Y}}\left(\frac{\bf{1}}{\bf{2}}\right)\right)\\
\vspace{0,5cm}
=&C_{\bf{X}}\left(\frac{\bf{1}}{\bf{2}}\right)+{\hat{C}}_{\bf{X}}\left(\frac{\bf{1}}{\bf{2}}\right)-\displaystyle\frac{1}{d+1}\left(C_{\bf{X}}\left(\frac{\bf{1}}{\bf{2}}\right)+{\hat{C}}_{\bf{X}}\left(\frac{\bf{1}}{\bf{2}}\right)\right)\\
\vspace{0,5cm}
&-\displaystyle\frac{1}{d+1}\sum_{j=1}^{d}\left(C_{\sigma_j\bf{X}}\left(\frac{\bf{1}}{\bf{2}}\right)+{\hat{C}}_{\sigma_j\bf{X}}\left(\frac{\bf{1}}{\bf{2}}\right)\right)\\
\vspace{0,5cm}
=&\displaystyle\frac{d}{d+1}\left(C_{\bf{X}}\left(\frac{\bf{1}}{\bf{2}}\right)+{\hat{C}}_{\bf{X}}\left(\frac{\bf{1}}{\bf{2}}\right)\right)
-\displaystyle\frac{d}{d+1}\displaystyle\frac{1}{d}\sum_{j=1}^{d}\left(C_{\sigma_j\bf{X}}\left(\frac{\bf{1}}{\bf{2}}\right)+{\hat{C}}_{\sigma_j\bf{X}}\left(\frac{\bf{1}}{\bf{2}}\right)\right),
\end{array}
\end{eqnarray}
that matches $\displaystyle\frac{d}{d+1}\beta({\bf{X}})$, applying again (\ref{R3}).
\end{proof}

\section{Application to real data}

The multivariate medial correlation coefficient in (\ref{defbetamulti}) can be estimated through the bivariate coefficients in (\ref{defbeta,i,D-i}). Here we consider the respective empirical counterparts. This estimation procedure has already been addressed in literature (Blomqvist \cite{Blom}, Schmid and Schmidt \cite{S+S} and references therein). 

Let $(X_{1,j},...,X_{d,j})$, $j=1,...,n$, be a random sample generated from $(X_{1},...,X_{d})$. Consider $$\hat{U}_{i,j}=\hat{F}_{X_i}(X_{i,j})=\frac{1}{n+1}\sum_{l=1}^{n}\mathds{1}_{\{X_{i,l}\leq X_{i,j}\}},\, i=1,...,d,\, j=1,...,n\,,$$
as well as, $\hat{M}_j\left(D\setminus \{i\}\right)=\bigvee_{r\in D\setminus \{i\}} \hat{U}_{r,j}$ and $\hat{W}_j\left(D\setminus \{i\}\right)=\bigwedge_{r\in D\setminus \{i\}} \hat{U}_{r,j}$. Based on (\ref{defbetamulti}) we define
\begin{eqnarray}\label{beta_estim}
\hat{\beta}= \frac{1}{d}\sum_{i=1}^{d}\hat{\beta}_{\{i\},D\setminus \{i\}},
\end{eqnarray}
where, according to (\ref{defbeta,i,D-i}), we take
$$
\hat{\beta}_{\{i\},D\setminus \{i\}}=\frac{\hat{\beta}\left(\hat{U}_{i},\hat{M}\left(D\setminus \{i\}\right)\right)+\hat{\beta}\left(\hat{U}_{i},\hat{W}\left(D\setminus \{i\}\right)\right)}{2},
$$
with
\begin{eqnarray}\nonumber
\begin{array}{rl}
&\hat{\beta}\left(\hat{U}_{i},\hat{M}\left(D\setminus \{i\}\right)\right)\\
=&\displaystyle 2\left(\frac{1}{n}\sum_{j=1}^{n}\left(\mathds{1}_{\{\hat{U}_{i,j}\leq 1/2\}}\mathds{1}_{\{\hat{M}_j(D\setminus \{i\})\leq 1/2\}}+\mathds{1}_{\{\hat{U}_{i,j}> 1/2\}}\mathds{1}_{\{\hat{M}_j(D\setminus \{i\})> 1/2\}}\right)\right)-1
\end{array}
\end{eqnarray}
and
\begin{eqnarray}\nonumber
\begin{array}{rl}
&\hat{\beta}\left(\hat{U}_{i},\hat{W}\left(D\setminus \{i\}\right)\right)\\
=&\displaystyle 2\left(\frac{1}{n}\sum_{j=1}^{n}\left(\mathds{1}_{\{\hat{U}_{i,j}\leq 1/2\}}\mathds{1}_{\{\hat{W}_j(D\setminus \{i\})\leq 1/2\}}+\mathds{1}_{\{\hat{U}_{i,j}> 1/2\}}\mathds{1}_{\{\hat{W}_j(D\setminus \{i\})> 1/2\}}\right)\right)-1.
\end{array}
\end{eqnarray}

We are going to apply the multivariate medial correlation coefficient estimator $\hat{\beta}$ in (\ref{beta_estim}) on two datasets. 

First, we consider the main GDP aggregates per capita in the European Union (EU), Germany and Portugal, available in \url{https://ec.europa.eu/eurostat/data/database}. We consider anual data from 2008 to 2019. The respective scatterplots are in Figure \ref{Fig1}. Germany and EU seem the most correlated. The estimates of the bivariate coefficients $\beta_{\{i\},D\setminus \{i\}}$ and of the multivariate medial correlation coefficient $\beta$ are in Table \ref{tab1}. We can see that the bivariate medial correlation between Portugal and the remaining EU and Germany presents the lowest contribution to the multivariate medial correlation.

\begin{center}
	\begin{figure}[h!]
		\includegraphics[width=5cm,height=5cm]{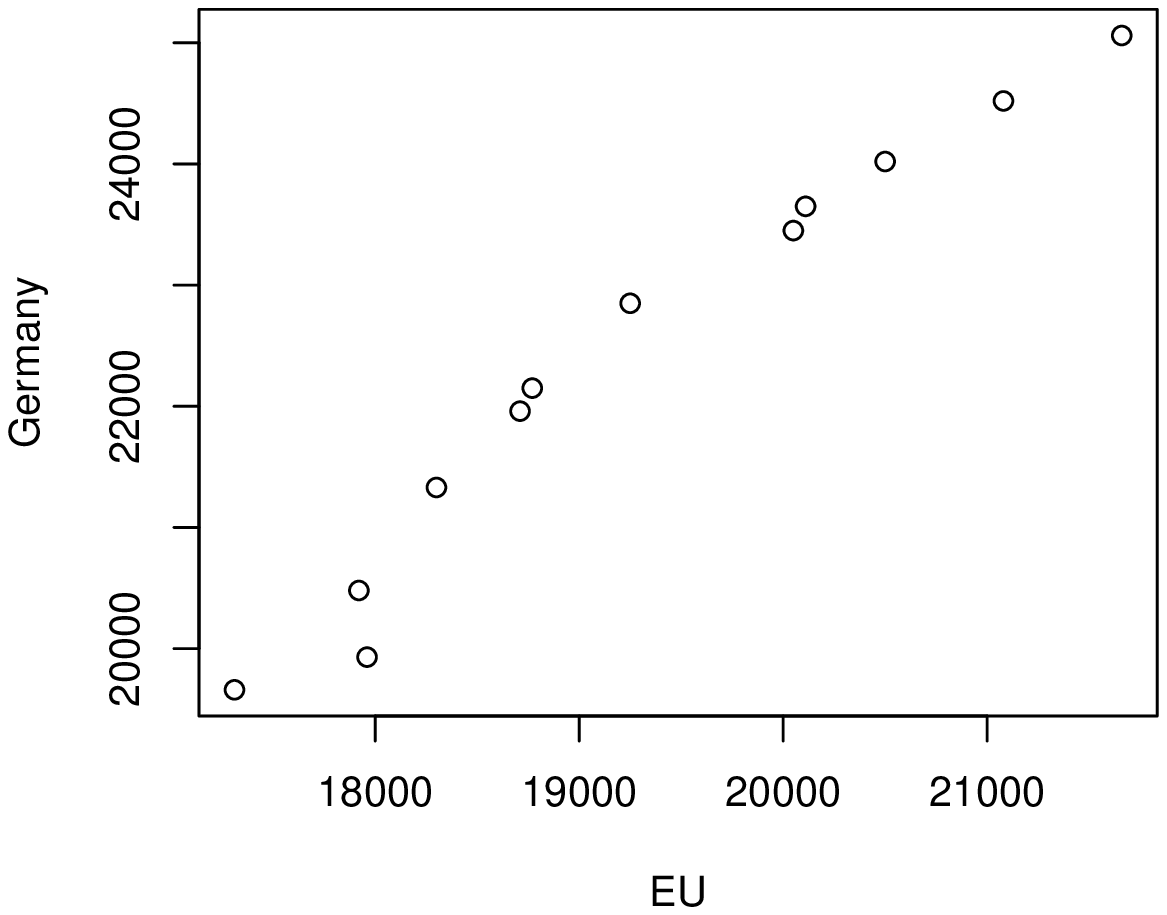}
		\includegraphics[width=5cm,height=5cm]{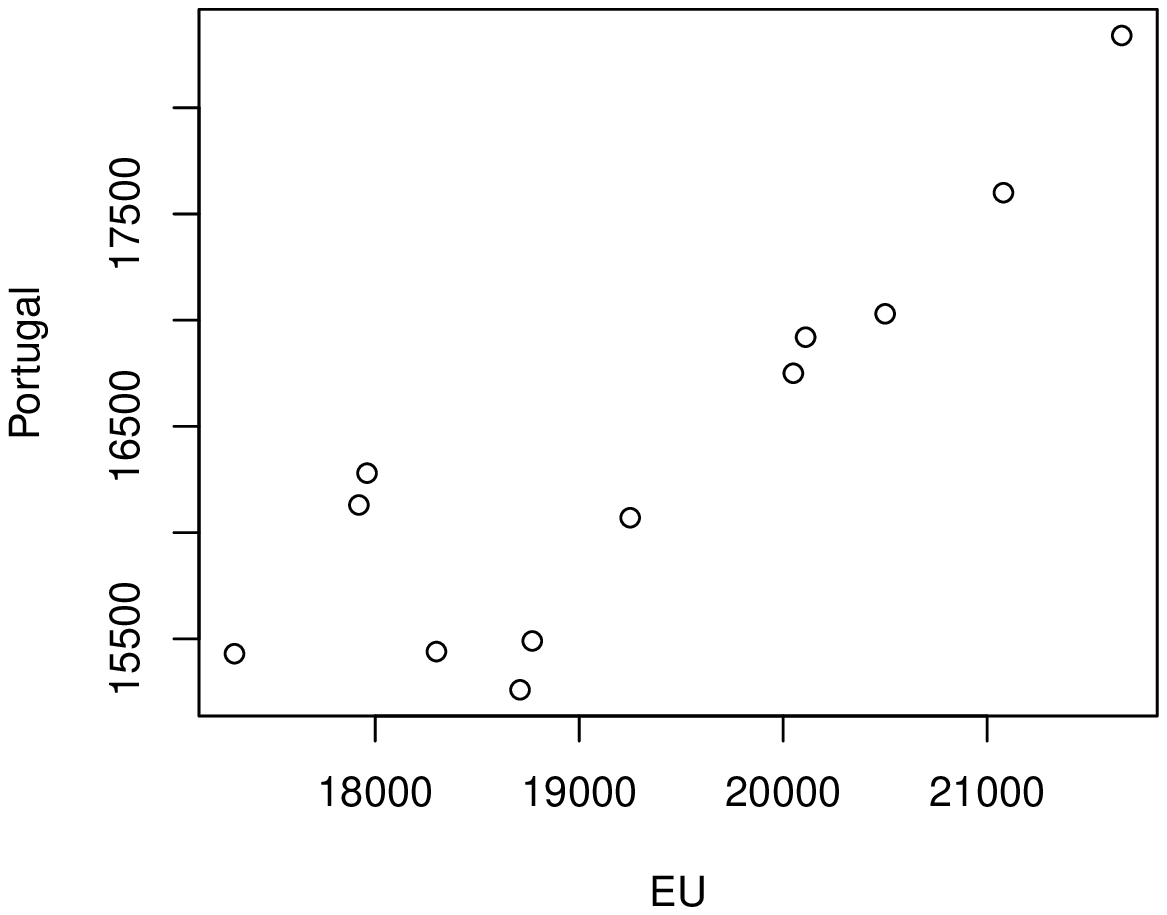}
		\includegraphics[width=5cm,height=5cm]{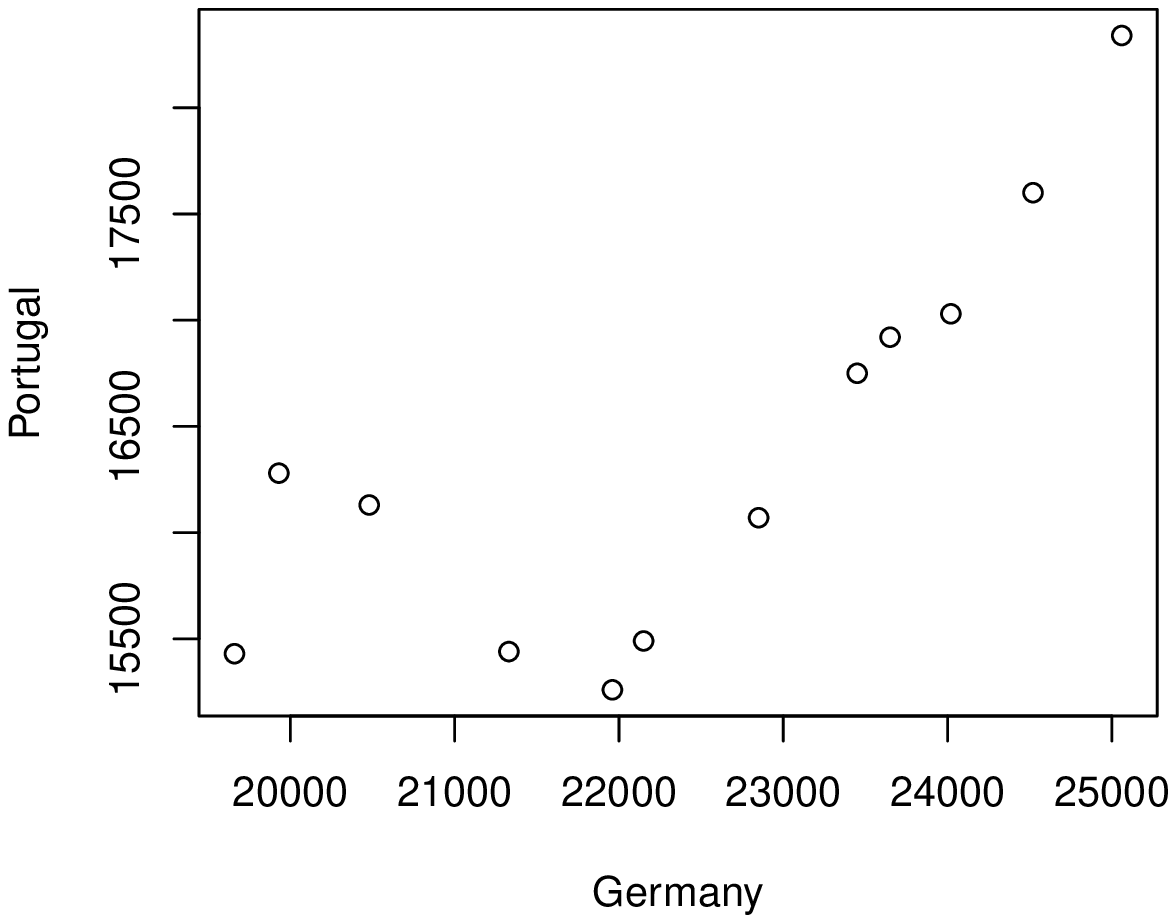}
		\caption{Anual main GDP aggregates per capita in the European Union versus Germany (left), European Union versus Portugal (center) and Germany versus Portugal (right). \label{Fig1}}
	\end{figure}
\end{center}

\begin{table}[h!]
	\begin{center}
		\caption{Estimates of the bivariate coefficients $\beta_{\{i\},D\setminus \{i\}}$ and of the multivariate medial correlation coefficient $\beta$ of the anual main GDP aggregates per capita in the European Union, Germany and Portugal, from 2008 to 2019. \label{tab1}}
		\begin{tabular}{c|c|c|c}
			$\{i\}$ & $D\setminus \{i\}$ & $\hat{\beta}_{\{i\},D\setminus \{i\}}$ & $\hat{\beta}$\\
			\hline
			\{EU\} & \{Germany, Portugal\} & 0.833 &\\
			\{Germany\} & \{EU, Portugal\} & 0.833 & 0.778\\
			\{Portugal\} & \{EU, Germany\} & 0.667 &\\
		\end{tabular}
	\end{center}
\end{table}

Now we consider a dataset related to white variants of the Portuguese ``Vinho Verde" wine, available in \url{http://archive.ics.uci.edu/ml/datasets/Wine+Quality}. See also Cortez \emph{et al.} (\cite{Cortez+}). Our analysis focuses on variables \textit{residual sugar}, \textit{density} and \textit{alcohol}, whose respective scatterplots are plotted in Figure \ref{Fig3}. It is visible some negative association between alcohol and density, as well as, between alcohol and residual sugar. On the other hand, density  and residual sugar are positively correlated. The estimates of the bivariate coefficients $\beta_{\{i\},D\setminus \{i\}}$ and of the multivariate medial correlation coefficient $\beta$ (Table \ref{tab3}) reflect this lack of concordance, with a larger negative  bivariate coefficient between alcohol and the remaining variables.

\begin{center}
	\begin{figure}[h!]
		\includegraphics[width=5cm,height=5cm]{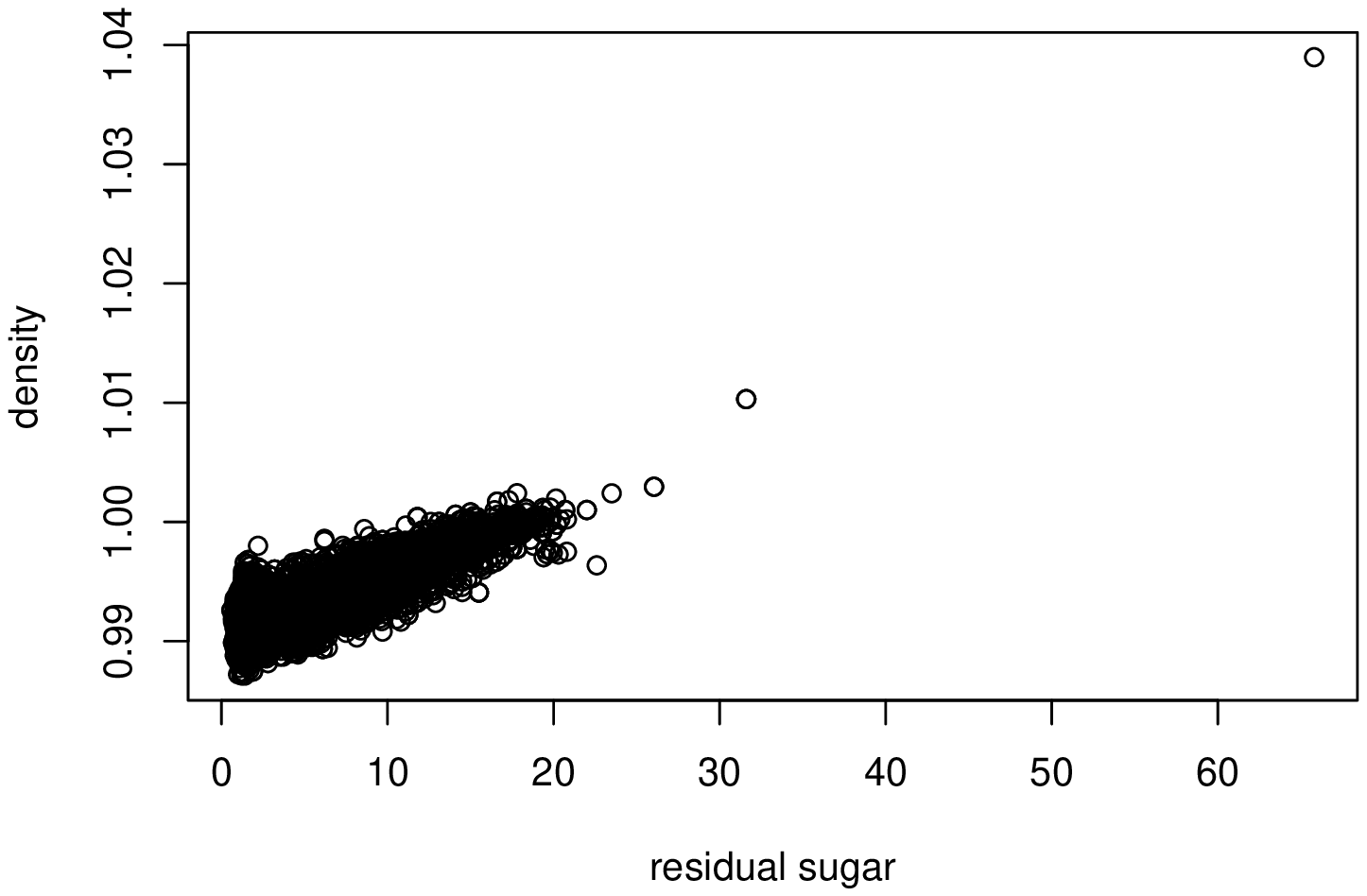}
		\includegraphics[width=5cm,height=5cm]{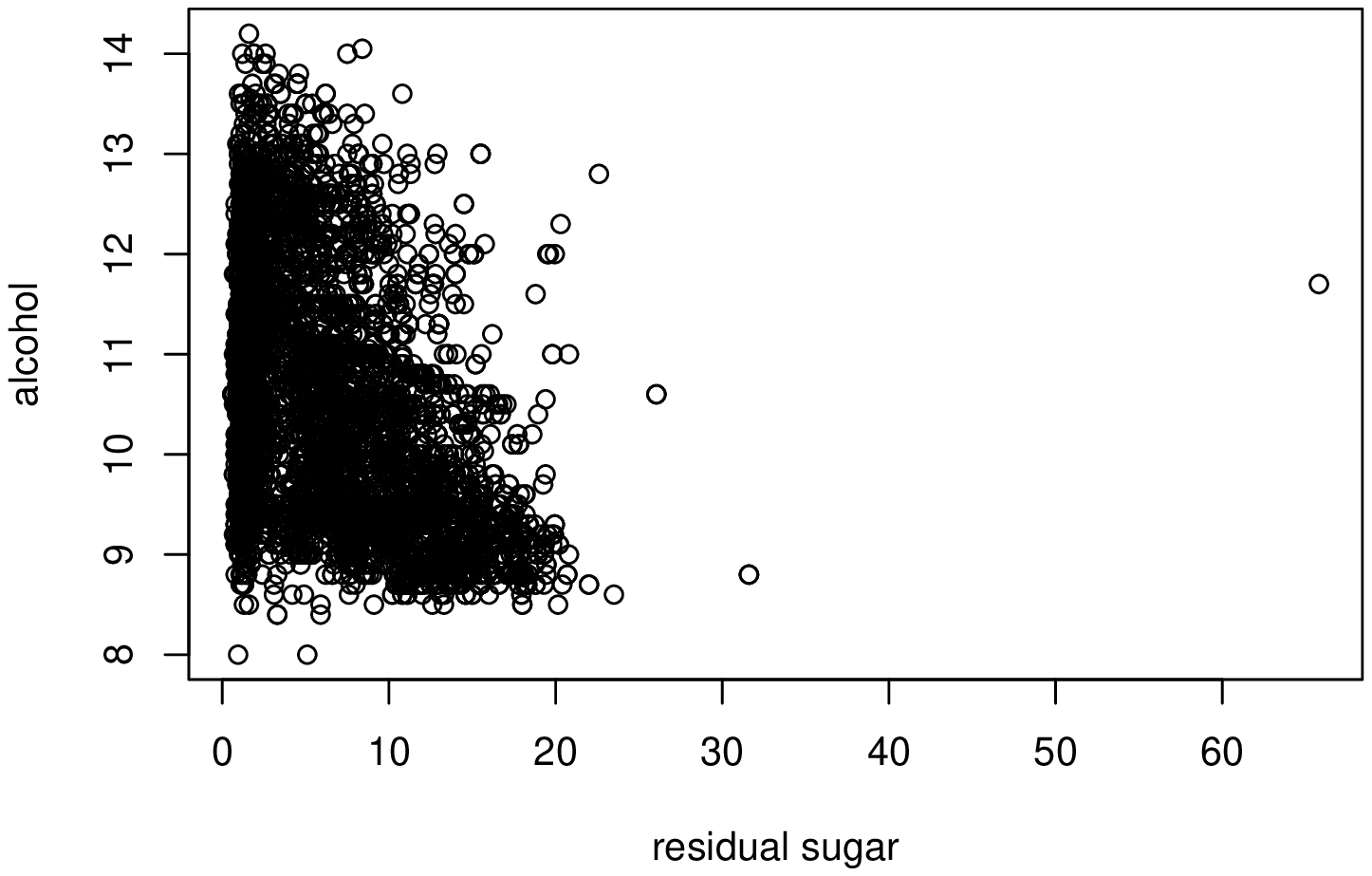}
		\includegraphics[width=5cm,height=5cm]{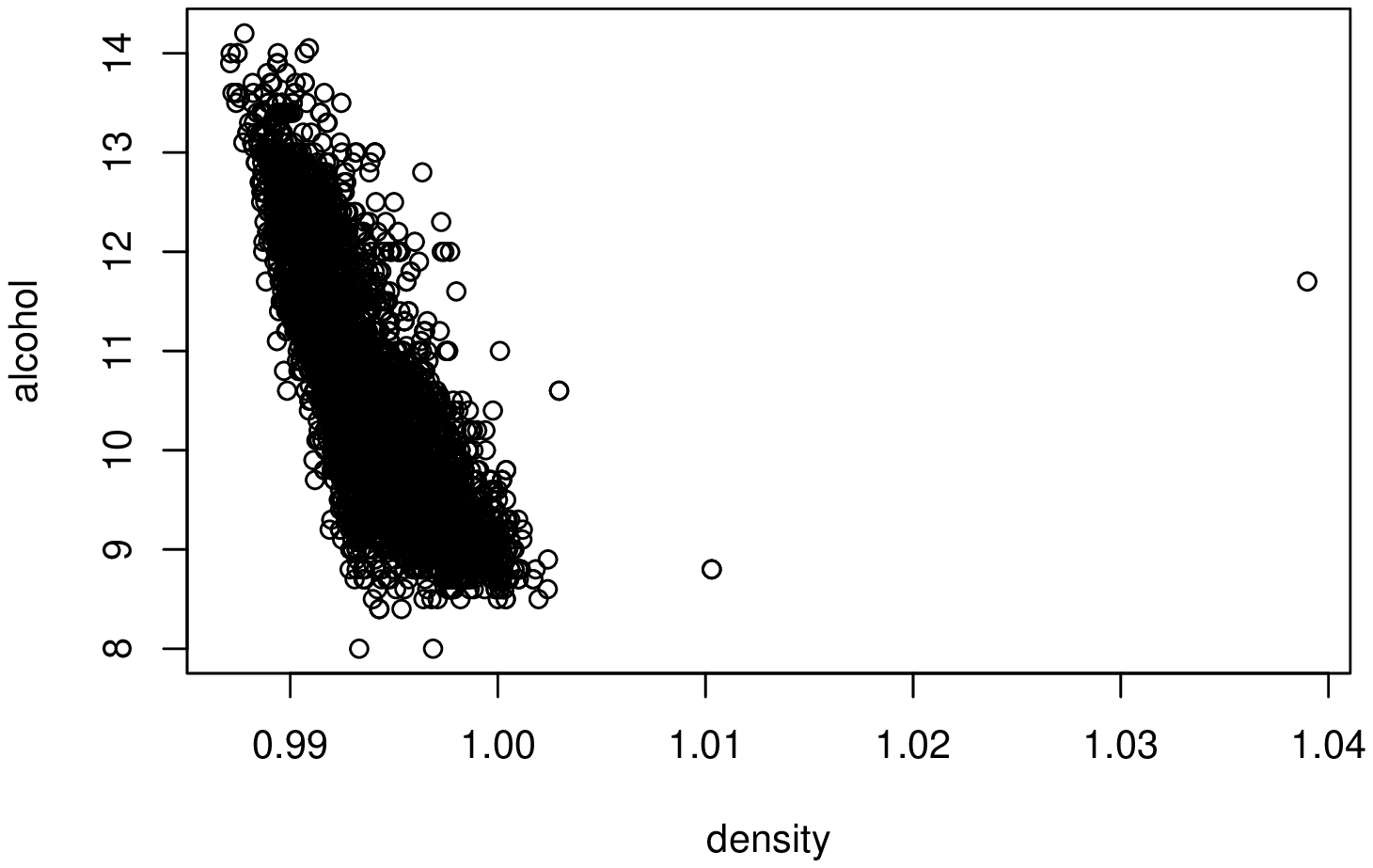}
		\caption{Scatterplots of  the variables \textit{residual sugar} versus \textit{density} (left), \textit{residual sugar} versus \textit{alcohol} (center) and \textit{density} versus \textit{alcohol} (right) within the wine dataset. \label{Fig3}}
	\end{figure}
\end{center}

\begin{table}[h!]
	\begin{center}
		\caption{Estimates of the bivariate coefficients $\beta_{\{i\},D\setminus \{i\}}$ and of the multivariate medial correlation coefficient $\beta$ for the variables \textit{residual sugar}, \textit{density} and \textit{alcohol} within the wine dataset.  \label{tab3}}
		\begin{tabular}{c|c|c|c}
			$\{i\}$ & $D\setminus \{i\}$ & $\hat{\beta}_{\{i\},D\setminus \{i\}}$ & $\hat{\beta}$\\
			\hline
			$\{\textrm{residual sugar}\}$ & $\{\textrm{density, alcohol}\}$ & 0.250 &\\
			$\{\textrm{density}\}$ & $\{\textrm{residual sugar, alcohol}\}$ & 0.179 & 0\\
			$\{\textrm{alcohol}\}$ & $\{\textrm{residual sugar, density}\}$ &-0.429 &\\
		\end{tabular}
	\end{center}
\end{table}

\section{Conclusion}

The multivariate medial correlation coefficient that we propose extends the probabilistic interpretation and properties of the Blomqvist $\beta$ coefficient, it is calculable from the copula, incorporates the dependence between each margin of the vector and the vector of the remaining margins and is a measure of a strong mode of multivariate concordance.

The estimation is addressed based on bivariate inferential methodology existing in literature and we illustrate its application using real data.

The adopted approach envisages the possibility of considering other functions of bivariate coefficients envolving  extremes of subvectors of ${\bf{X}}$, as well as the possibility of adapting the method to generalize other coefficients of bivariate dependence.

\section*{Acknowledgements}

The authors thank the reviewers and the associated editor for the very important and valuable comments that contributed to the improvement of this work. 

\noindent The first author was partially supported by the research unit Centre of Mathematics and Applications of University of Beira Interior
UIDB/00212/2020 - FCT (Funda\c c\~ao para a Ci\^encia e a Tecnologia).  
The second author was financed by Portuguese Funds through FCT - Fundação para a Ciência e a Tecnologia within the Projects UIDB/00013/2020 and UIDP/00013/2020 of Centre of Mathematics of the University of Minho, UIDB/00006/2020 of Centre of Statistics and its Applications of University of Lisbon and PTDC/MAT-STA/28243/2017.

%------------------------------------BIBLIOGRAFIA----------------

\end{document}